\documentclass[letterpaper, 10 pt, conference]{ieeeconf} 

\usepackage{graphics}      
\usepackage{amsmath,amsfonts,amssymb}
\usepackage{pstricks,xcolor}
\usepackage{array}
\usepackage{epsfig} 
\graphicspath{{./figures/}}
\usepackage{subcaption}
\usepackage{multicol}

\IEEEoverridecommandlockouts                              

\newtheorem{theorem}{Theorem}[section]
\newtheorem{lemma}[theorem]{Lemma}

\newtheorem{rmk}[theorem]{Remark}
\newtheorem{corollary}[theorem]{Corollary} 
\newtheorem{defn}[theorem]{Definition} 

\title{\LARGE \bf
Fault Tolerant Control for Networked Mobile Robots
}

\author{Pietro Pierpaoli$^{1}$, Dominique Sauter$^{2}$, and Magnus Egerstedt$^{1}$
\thanks{$^{1}$Pietro Pierpaoli and Magnus Egerstedt are with School of Electrical and Computer Engineering,
        Georgia Institute of Technology, Atlanta, GA 30332 USA
        {\tt\small \{pietro.pierpaoli,magnus\}@gatech.edu }}%
\thanks{$^{2}$ CRAN CNRS-UMR 7039, Université de Lorraine, Vandoeuvre-les-Nancy, France 
        {\tt\small dominique.sauter@univ-lorraine.fr}}%
}

\begin{document}

\maketitle
\thispagestyle{empty}
\pagestyle{empty}

\begin{abstract}                
Teams of networked autonomous agents have been used in a number of applications, such as mobile sensor networks and intelligent transportation systems. However, in such systems, the effect of faults and errors in one or more of the sub-systems can easily spread throughout the network, quickly degrading the performance of the entire system. In consensus-driven dynamics, the effects of faults are particularly relevant because of the presence of unconstrained rigid modes in the transfer function of the system. Here, we propose a two-stage technique for the identification and accommodation of a biased-measurements agent, in a network of mobile robots with time invariant interaction topology. We assume these interactions to only take place in the form of relative position measurements. A fault identification filter deployed on a single observer agent is used to estimate a single fault occurring anywhere in the network. Once the fault is detected, an optimal leader-based accommodation strategy is initiated. Results are presented by means of numerical simulations and robot experiments.
\end{abstract}


\section{Introduction}
\label{sec:introduction}
Cooperation in multi-agent systems can lead to highly coordinated behaviors even when individual agents in the network have limited skills and information. For instance, a team of cooperating robots can efficiently adapt to changes in the environment and perform sequential tasks, with potentially limited supervision. The degree of interaction depends on the sensor and communication architectures, and can be conveniently modeled using elements from graph 
theory~\cite{olfati2006flocking},~\cite{mesbahi2010graph}.

As an example, {\it consensus theory} constitutes a fundamental tool for the design of control and estimation protocols in multi-agent 
systems~\cite{garin2010survey}. 
Despite their efficiency and versatility, consensus-based algorithms are vulnerable to drift, and therefore, measurements errors and hardware faults can introduce disturbances that are difficult to correct~\cite{xiao2007distributed}.

This paper addresses this problem in the context of Fault Tolerant Control Systems (FTCS). The study of FTCS addresses the design of active control systems capable of automatically detecting a fault and performing the actions required in order to maintain acceptable 
performances~\cite{zhang2008bibliographical}. FTCS design inherently involves a multi-stage process: first, a Fault Detection and Identification (FDI) system must provide precise information about the fault; then, during the Fault Accommodation (FA) (or mitigation) stage, an appropriate control compensates for the fault. 

In this paper we describe a Fault Identification Filter (FIF) to be deployed on multi-agent robotic systems, performing linear agreement ({\it consensus}) and formation control protocols. The proposed technique relies on a bank of linear observers deployed on board of an {\it observer} agent in charge of detecting a single, time-invariant fault occurring in any of the nodes of the network, including itself. Once the fault is detected, a leader, not necessarily corresponding to the observer, compensates for it. We consider the underlying graph representing the interaction between the agents to be fixed in time and completely controllable with respect to the leader node; we also assume robots interactions to occur as relative distance measurements only.

\subsection{Related Work}
Many different approaches have been proposed for the design of fault tolerant distributed Network Control Systems (NCS) (see for example~\cite{patton2007generic} and references therein). In the context of distributed networked agents equipped with sensors, both noise and faults make measurements unreliable, which has been approached using distributed Kalman filtering~\cite{olfati2007distributed}, Bayesian~\cite{luo2006distributed} and Dempster-Shefer frameworks~\cite{premaratne2009dempster}. When the states of the agents directly depends on neighbors' relative states, the ability to restore the nominal  pre-fault performance, becomes more complex, and depends on the underlying interaction graph between the nodes. To this end, the problem of networked control systems usually targets the selection of valid control nodes and the study of the interaction betwen topology and controllability~\cite{liu2011controllability},~\cite{pasqualetti2014controllability},~\cite{chapman2013semi}. 

In robotics, fault tolerant controls have been investigated in many different contexts, including, for example, flight control systems~\cite{steinberg2005historical}, manipulators~\cite{visinsky1995dynamic}, and quadrocopters subject to the loss of motors~\cite{mueller2014stability}. Resilience of multi-agent robotic systems to errors, faults, and adversarial attacks has also been investigated. In~\cite{bandyopadhyay2017probabilistic}, fault-resistance is achieved by assuming agents controlled by independent probabilistic processes. However, when agents' controllers depend on neighbors' state, the spreading of faults and errors becomes difficult to control. 

In~\cite{shames2011distributed} 
a filter composed by a bank of linear observers, where each observer can detect faults in one of its nearest neighbors is proposed. After fault occurrence, the faulty node is assumed to be removed from the network. Heterogeneous multi-agent systems are considered 
in~\cite{davoodi2014distributed}, 
where a LMI-based solution to the distributed FDI problem is proposed. Multiple and simultaneous faults occurring in each agent and its nearest neighbors were detected considering environmental noise and disturbances. 
In~\cite{pasqualetti2012consensus}, 
{\it misbehaving} nodes are studied in the context of linear consensus dynamics, where both genuine random faults and malicious messages are considered. The fault identification was then studied in terms of the connectivity properties of the network.

The formulation used in this paper builds on the techniques described in~\cite{liu1997fault} and~\cite{keller1999fault}, in which a linear observer is designed such that the filter residuals possess some desired directional properties. The original idea was extended to linear networked control systems with communication delays in~\cite{sauter2009robust} and to discrete-time switched linear systems in~\cite{rodrigues2006fault}. 

In Section~\ref{sec:fif} we describe the dynamics of the robots and the fault detection filter. In Section~\ref{sec:optimalAccom.}, we discuss the leader optimal FA strategy, presenting numerical simulations for a multi-agent robotic system. In Section~\ref{sec:formation}, we extend the results to a formation control protocol, which was tested on real robots. Final remarks are reported in Section~\ref{sec:conc}.
\section{Fault Identification Filter Under Pure Consensus Dynamics}
\label{sec:fif}
\subsection{System Dynamics and Fault Modeling}
Consider a collection of $n$ mobile robots, located in a planar connected, and compact domain $\mathcal{D}\subset\mathbb{R}^2$. Let $x_i \in \mathbb{R}^2$, for $i=1\dots,n$, denote the $i^{\text{th}}$ agent's position. We assume each agent to interact with a non-empty set of other robots. Generally, inter-agent interactions can be conveniently described by an undirected graph $\mathcal{G}=(V,E)$, where $V$ is the set of $n$ nodes representing the agents, and $E$ is an unordered list of node pairs corresponding to interacting robots. We assume undirected networks, i.e., $(i,j)=(j,i)$ and we let $\mathcal{G}$ be connected and constant at all times. 

The {\it neighborhood set} of agent $i$, denoted by $\mathcal{N}_i$, $i=1,\dots,n$, is the set of all vertexes connected to node $i$. Also, $m_i$ is the {\it degree} of vertex $i$, defined as the number of vertexes connected to node $i$, i.e. $m_i=|\mathcal{N}_i|$.

To start the development, we initially consider a team of robots performing a consensus protocol, which represents a general starting point for other more complex possible behaviors. Assuming a discrete time model, with positive step size $\epsilon$, and temporally indexed by $k$, the update equation for agent $i$ is
\begin{equation} \label{eq:consensus}
x_i(k+1) = x_i(k)-\epsilon\sum_{j \in \mathcal{N}_i} (x_i(k)-x_j(k)) \quad i={1,\dots,n}.
\end{equation}
The behavior emerging from the dynamics in~(\ref{eq:consensus}) is a pure consensus dynamics, leading to the rendezvous of the agents at the centroid of their initial configuration (see for example~\cite{mesbahi2010graph} and references therein).

We write the complete state of the system in the compact form ${\bf x}=[x_1^T,\dots,x_n^T]^T$, where ${\bf x}\in\mathbb{R}^{2n}$. It is possible to represent the complete update equation as:
\begin{equation}\label{eq:consensusCompact}
{\bf x}(k+1) = A{\bf x}(k) 
\end{equation}
where, $A=( I_{2n}-\epsilon L\otimes I_2)\in\mathbb{R}^{2n \times 2n}$, $\otimes$ is the Kronecker product, $I_p$ the identity matrix of size $p\times p$, and $L\in\mathbb{R}^{n \times n}$ is the Laplacian of $\mathcal{G}$. 

Now, assume that at time step $k_d>0$, an a-priori unknown agent, indexed by $i^*$ experiences a fault. Without loss of generality, we model this fault as an exogenous velocity input $\delta\in\mathbb{R}^2$ acting on the system; therefore, including the fault in the state update equation~(\ref{eq:consensusCompact}) leads to the switched system:
\begin{align} \label{eq:switchedsystem}
{\bf x}(k+1) &= A{\bf x}(k) + \epsilon F_{i^*} \, \delta \, \nu(k) \\
\nu(k) &=
\begin{cases}
1 \,\, &\text{if} \,\, k \geq k_d \\
0 \,\, &\text{if} \,\, k <    k_d, \\
\end{cases}
\end{align}
where $F_{i^*}\in\mathbb{R}^{2n\times 2}$ is the fault distribution matrix corresponding to the unknown faulty agent, and defined as:
\begin{equation*}
F_{i^*} = e_{i^*} \otimes I_2,
\end{equation*} 
where $e_{i^*}$ is the $i^{* \text{th}}$ canonical vector of appropriate size. 

A possible interpretation for the fault $\delta$ can be given as follows. The interaction between the agents occurs in the form of relative distance measurements; therefore, the measurements for agent $i$, $y_{i}(k)\in\mathbb{R}^{2m_{i}}$, $i=1,\dots,n$, are:
\begin{equation} \label{eq:measure}
y_{i}(k) = \begin{bmatrix}
x_{i}(k) - x_{j_1}(k) \\
\vdots \\
x_{i}(k) - x_{j_{m_{i}}}(k) \end{bmatrix} = C_{i}{\bf x}(k),
\end{equation}
where the indexes $j_1,\dots,j_{m_{i}}$ indicate the neighbors of agent $i$ and $C_i\in\mathbb{R}^{2m_i \times 2n}$ is the corresponding measure matrix. The state update equation can then be written in terms of the measurements as:
\begin{equation} \label{eq:vel_measure}
x_{i}(k+1) =x_{i}(k) -\epsilon({\bf 1}_{m_{i}}^T\otimes I_2)\,y_{i}(k),
\end{equation}
where ${\bf 1}_m\in\mathbb{R}^m$ is a vector having all its elements equal to $1$. At time step $k_d$, agent $i^*$ experiences a fault in one or more of its sensors. Denoting by $m_{i^*}$ the degree of node $i^*$, we assume each of the post-fault measurements to be biased by a non-homogeneous term $\delta_{i^*j}\in\mathbb{R}^2$, with $j=1,\dots, m_{i^*}$. Then, the vector of measurements for the faulty agent is:
\begin{equation} \label{eq:compactFault}
y_{i^*}(k) = \begin{bmatrix}
x_{i^*}(k) - x_{j_1}(k) + \delta_{i^*j_1} \\
\vdots \\
x_{i^*}(k) - x_{j_{m_{i^*}}}(k) + \delta_{i^*j_{m_{i^*}}} \end{bmatrix} = C_{i^*}\,{\bf x} + \hat{\delta}
\end{equation}
where $ \hat{\delta}=[ \delta_{i^*j_1}^T,\dots,\delta_{i^*j_{m_{i^*}}}^T]^T $, is the compact representation of all faults corresponding to each neighbor in $\mathcal{N}_{i^*}$. By applying~(\ref{eq:compactFault}) to~(\ref{eq:vel_measure}), by inspection we note that the fault acting on the system in~(\ref{eq:switchedsystem}) corresponds in this case to:
\begin{equation*}
\delta=-({\bf 1}_{m_{i^*}}^T\otimes I_2)\,\hat{\delta}.
\end{equation*}

\subsection{Fault Observer Design}
We now turn our discussion to the design of the filters deployed on the one agent of the team in charge of detecting and estimating a fault in any robot of the team, including itself. We refer to this agent as the {\it observer agent}, and we denote quantities relative to it with a subscript $o$; e.g., $C_o$ represents the observer agent's measurements matrix. 

\begin{rmk}
The purpose of the observer $o$ is to:
\begin{itemize}
\item find the index $i^*$ corresponding to the  faulty agent;
\item estimate the fault vector $\delta$;
\item find the time of fault occurrence $k_d$.
\end{itemize}
\end{rmk}

With reference to Fig.~\ref{fig:filterbank}, consider a bank of $n$ filters, each designed to detect a fault in a specific agent of the team. Note that by defining $n$ filters, we allow the observer to be the faulty agent and being capable of detecting itself as the faulty agent. The common input for each filter in the bank is the observer measure $y_o(k)$, while each filter outputs a set of signals called residuals. When residuals are sensitive to only a single fault the design belongs to the category of Fault Isolation Filter (FIF). We denote quantities associated with each filter of the bank by the subscript $i$.

\begin{figure}[h]

\psscalebox{0.5 0.5} 
{
\begin{pspicture}(-3,-2)(8,4.2)
\rput(0,0){\psframe[linecolor=black, linewidth=0.04, dimen=outer,linestyle=dashed,dash=3pt 2pt](12,3)(-1,-1.4)}

\rput(-0.5,0){\psframe[linecolor=black, linewidth=0.04, dimen=outer](3,2)(0,0)
\psline[linecolor=black, linewidth=0.04, arrowsize=0.05cm 2.0,arrowlength=1.4,arrowinset=0.0]{->}(0.7,0)(0.7,-0.5)
\psline[linecolor=black, linewidth=0.04, arrowsize=0.05cm 2.0,arrowlength=1.4,arrowinset=0.0]{->}(2.2,0)(2.2,-0.5)
\large \rput[bl](0.5,-1){$\alpha_1(k)$}
\rput[bl](2.0,-1){$\gamma_1(k)$}
\LARGE \rput[bl](0.8,0.6){$K_1$}
}
\rput(3.,0){\psframe[linecolor=black, linewidth=0.04, dimen=outer](3,2)(0,0)
\psline[linecolor=black, linewidth=0.04, arrowsize=0.05cm 2.0,arrowlength=1.4,arrowinset=0.0]{->}(0.7,0)(0.7,-0.5)
\psline[linecolor=black, linewidth=0.04, arrowsize=0.05cm 2.0,arrowlength=1.4,arrowinset=0.0]{->}(2.2,0)(2.2,-0.5)
\large \rput[bl](0.5,-1){$\alpha_2(k)$}
\rput[bl](2.0,-1){$\gamma_2(k)$}
\LARGE \rput[bl](0.8,0.6){$K_2$}
}
\rput(8,0){\psframe[linecolor=black, linewidth=0.04, dimen=outer](3,2)(0,0)
\psline[linecolor=black, linewidth=0.04, arrowsize=0.05cm 2.0,arrowlength=1.4,arrowinset=0.0]{->}(0.7,0)(0.7,-0.5)
\psline[linecolor=black, linewidth=0.04, arrowsize=0.05cm 2.0,arrowlength=1.4,arrowinset=0.0]{->}(2.2,0)(2.2,-0.5)
\large \rput[bl](0.5,-1){$\alpha_n(k)$}
\rput[bl](2.0,-1){$\gamma_n(k)$}
\LARGE \rput[bl](0.8,0.6){$K_n$}
}

\psline[linecolor=black, linewidth=0.04, arrowsize=0.05cm 2.0,arrowlength=1.4,arrowinset=0.0]{->}(6,3.8)(6,2.6)
\psline[linecolor=black, linewidth=0.04](1,2.5)(9.5,2.5)
\psline[linecolor=black, linewidth=0.04, arrowsize=0.05cm 2.0,arrowlength=1.4,arrowinset=0.0]{->}(1,2.5)(1,2)
\psline[linecolor=black, linewidth=0.04, arrowsize=0.05cm 2.0,arrowlength=1.4,arrowinset=0.0]{->}(4.5,2.5)(4.5,2)
\psline[linecolor=black, linewidth=0.04, arrowsize=0.05cm 2.0,arrowlength=1.4,arrowinset=0.0]{->}(9.5,2.5)(9.5,2)


\large \rput[bl](6,4){$y_o(k)$}
\rput[bl](6.8,1){$\dots$}
\rput[bl](-0.5,3.2){Fault Identification Filter}
\end{pspicture}

}
	\caption{Bank of fault identification filters deployed on the observer. Common input is the observer's measure $y_o(k)$ and outputs are $2n$ residuals. $K_i$ is the gain of the $i^{\text{th}}$ filter. \label{fig:filterbank}}
\end{figure}
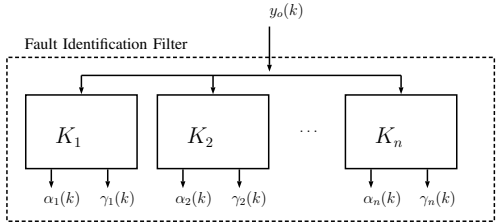

Assuming a linear state observer for the dynamics in~(\ref{eq:consensusCompact}), the update equation for the state estimate in the $i^{th}$ filter is:
\begin{eqnarray}\label{eq:genFilter}
\begin{aligned}
\hat{{\bf x}}_o^i(k+1) &= (A-K_i C_o)\hat{{\bf x}}_o^i (k) + K_i\,y_o(k) \\
\hat{y}^i_o(k) &= C_o\,\hat{{\bf x}}_o^i(k). 
\end{aligned}
\end{eqnarray}

Following the approach described in \cite{keller1999fault}, we derive a formulation for the gain $K_i$ such that the filter outputs have some desired directional properties. To this end, the estimation error in filter $i$ at time $k$ is $e_i(k) = {\bf x}(k)-\hat{\bf x}_o(k)$, and under the effect of the fault we have:
\begin{eqnarray}\label{eq:errorFilter}
\begin{aligned}
e_i(k+1) &= (A-K_i C_o)e_i(k) + \epsilon F_i \, \delta \\
q_i(k) &= C_o\,e_i(k),
\end{aligned} 
\end{eqnarray}
where $q_i(k)$ is the residual relative to the filter $i$. In a similar way, denoting with $\bar{e}_i(k)$ and $\bar{q}_i(k)$ $i^{\text{th}}$ filter's estimate error and the residual for the system not subject to the fault respectively, we have:
\begin{eqnarray}\label{eq:errorFilter-nofault}
\begin{aligned}
\bar{e}_i(k+1) &= (A-K_i C_o)\bar{e}_i(k) \\
\bar{q}_i(k) &= C_o\,\bar{e}_i(k).
\end{aligned} 
\end{eqnarray}

We now introduce the concept of {\it fault detectability index} required for the definition of the FIF.
\begin{defn}\label{def:rho}
For the linear time invariant system~(\ref{eq:switchedsystem}) the {\it fault detectability index} $\rho_{i^*}\in\mathbb{R}$ is
\begin{equation}
\rho_{i^*} = \min \{ v\,:\, C_o A^{v-1}F_{i^*} \neq 0, \quad v=1,2,\dots \}.
\end{equation}
as defined in~\cite{liu1997fault}.
\end{defn}

If $\mathcal{G}$ is connected, the system has finite fault detectability index. Before proving this result we introduce the following lemma, by slightly reformulating the result in~\cite{olfati2007consensus}.
\begin{lemma}\label{lem:stochastic}
If $\mathcal{G}=(V,E)$ is connected and $\epsilon\leq\frac{1}{\max(m_i)}$, for $i=1,\dots,n$, then the matrix $A$ in~(\ref{eq:consensusCompact}) is stochastic.
\end{lemma}
\begin{proof} 
The matrix $A\in\mathbb{R}^{2n\times 2n}$ is row stochastic if its entries satisfy $0\leq A_{ij} \leq 1$, for all $i,j=1,...,2n$ and $\sum_{j=1}^{2n} A_{ij}=1$, for all $i=1,\dots,2n$. From the definition of the Laplacian of $\mathcal{G}$, $\sum_{j=1}^n -\epsilon L_{ij}=0$, for all $i=1,\dots,n$. Then, since $A=( I_{2n}-\epsilon L\otimes I_2)$, we note that $\sum_{j=1}^{2n} A_{ij} = 1- \sum_{j=1}^{2n} \epsilon(L\otimes I_2)_{ij}=1$, for all $i=1,\dots,2n$.

Again, by definition of the graph Laplacian, we have $L_{ii}=m_i$ (where $m_i\geq 1$ since $\mathcal{G}$ is connected), $L_{ij}=-1$ if $(i,j)\in E$, and $L_{ij}=0$ otherwise. Then, since the only non-null elements of $A$ are $A_{(2i-1)(2i-1)}=A_{(2i)(2i)}=1-\epsilon\, m_i$, for $i=1,\dots,n$, and $A_{(2i-1)(2j-1)}=A_{(2i)(2j)}=\epsilon$, for $(i,j)\in E$, $A$ is stochastic if conditions: 
\begin{equation}\label{eq:sthocCond}
0 \leq 1-\epsilon\,m_i \leq 1 \qquad
0 \leq \epsilon \leq 1
\end{equation}
are satisfied. Finally, since $m_i\geq 1$ and $\epsilon>0$, both~(\ref{eq:sthocCond}) hold if and only if $\epsilon\leq\frac{1}{\max(m_i)}$.
\end{proof}


For the graph $\mathcal{G}$, the geodesic distance between a pair of nodes is the length of the shortest path connecting them. We introduce the geodesic function $g: V\times V \rightarrow \mathbb{N}$ and $g_{i,j}$ is the geodesic distance between nodes $i$ and $j$.
\begin{theorem}\label{th:geodesic}
Under the hypotheses of Lemma~\ref{lem:stochastic}, the fault detectability index corresponds to the geodesic distance, i.e. $\rho_{i^*}=g_{o,i^*}$. 
\end{theorem}
\begin{proof}
Consider a discrete time random walk on $\mathcal{G}$ governed by the stochastic transition matrix $A$. Denoting with $\hat{F}(0)\in\mathbb{R}^{2n\times2}$ the initial probability distribution of two independent processes over $\mathcal{G}$, we note that $\hat{F}(v)=~[\hat{F}_1(v), \hat{F}_2(v)]=A^v\hat{F}(0)$ represents the probability distribution of the walks at time step $v$. Then, by definition of $F_{i^*}=e_{i^*}\otimes I_2$, the process described by $\hat{F}(v)=A^vF_{i^*}$, can also be interpreted as the probability distributions of two identical walks at step $v$, both started at node $i^*$. 

After $v=g_{o,i^*}-1$ steps, the probability that the walk (both walks are identical) reached node $o$ is zero, and therefore $\hat{F}_{o\ell}(g_{o,i^*}-1)=0$. Similarly, after the same number of steps, it is easy to verify that:
\begin{equation} \label{eq:minngb}
\exists\,j\in\mathcal{N}_o \quad g_{o,i^*}-1=\min\{v: \hat{F}_{j\ell}(v) \neq 0 \},
\end{equation}
where $\ell=1,2$ is the index of the walks. From the definition of the measurements matrix $C_o$ in~(\ref{eq:measure}), for walk $\ell$ we write:
\begin{equation*}
C_o \hat{F}_\ell(g_{o,i^*}-1) = \begin{bmatrix}
\hat{F}_{o\ell}(g_{o,i^*}-1) - \hat{F}_{j_1\ell}(g_{o,i^*}-1) \\
\vdots \\
\hat{F}_{o\ell}(g_{o,i^*}-1) - \hat{F}_{j_{m_{o}}\ell}(g_{o,i^*}-1) \end{bmatrix}
\end{equation*}
where, similarly to~(\ref{eq:measure}), indexes $j_1,\dots,j_{m_{o}}$ correspond to the neighbors of agent $o$. Finally, since $\hat{F}_{o\ell}(g_{o,i^*}-1)=0$ and~(\ref{eq:minngb}) we conclude that:
\begin{equation}
\min\{v: C_o A^{v-1}F_{i^*} \neq 0\}=g_{o,i^*},
\end{equation} 
and from Definition~\ref{def:rho}, it follows that $\rho_{i^*}=g_{o,i^*}$.
\end{proof}

In other words, the fault detectability index can be viewed as the number of steps required for the fault to affect an observer's neighbor and therefore, being visible to the observer itself. To this end, fault detectability in a network can be studied similarly to its controllability~\cite{yaziciouglu2016graph}.

\begin{corollary}
Under the hypothesis of Lemma~\ref{lem:stochastic}, for every choice of observer and faulty agent, a finite fault detectability index always exists.
\end{corollary}
\begin{proof}
For all connected graphs there exists a finite geodesic distance between each pair of nodes. Thus, this result directly follows from Theorem~\ref{th:geodesic}. 
\end{proof}	

\begin{defn}
The fault detectability matrix for the filter $i$, namely $D_i\in\mathbb{R}^{2m\times 2}$, is defined as $D_i = C_o\,\Psi_i$, where $\Psi_i = A^{\rho_i-1} \epsilon F_i$.
\end{defn}

The following is a main result from~\cite{keller1999fault}.
\begin{theorem} \label{th:keller}
Assume the following parametrization:
\begin{equation}\label{eq:K}
K_i = \omega_i\Pi_i + \bar{K}_i\Sigma_i
\end{equation}
with $\omega_i=A \Psi_i $, $\Pi_i=(D_i)^+$ \footnote{$A^+$ is the pseudoinverse (or Moore-Penrose inverse) of the matrix $A$}, $\Sigma_i=\beta_i(I_{m_o}-D_i\Pi_i)$, where $\omega_i\in\mathbb{R}^{2n\times 2}$, $\Pi_i\in\mathbb{R}^{2 \times 2m_o}$ and $\beta_i\in\mathbb{R}^{2m-2\times 2m}$ is an arbitrarily chosen matrix such that the matrix $\Sigma_i~\in~\mathbb{R}^{2m-2\times 2m}$ has full row rank; then, the following constraint is always satisfied:
\begin{equation*}
(A-K_iC_o)\Psi_i = 0
\end{equation*}
and we can write the residual at time $k$ as:
\begin{equation} \label{eq:qk}
q_i(k) = \bar{q}_i(k)+ D_i\delta(k-\rho_i)
\end{equation}
\end{theorem}

We refer the reader to~\cite{keller1999fault} for the details of the proof. Thanks to the particular parametrization introduced in Theorem~\ref{th:keller}, for the system affected by the fault, the residual $q_i(k)$ in~(\ref{eq:qk}) is given by the sum of two terms. The first term is the residual of the fault-free system~(\ref{eq:errorFilter-nofault}), while the second term depends on the fault vector $\delta$ delayed by a quantity dependent on the fault detectability index $\rho_i$.

Following the results from Theorem~\ref{th:keller}, substituting~(\ref{eq:K}) in~(\ref{eq:genFilter}) leads to the following final expression for the fault identification filter $i$:
\begin{equation}
\hat{{\bf x}}_o^i(k+1) = A\hat{{\bf x}}_o^i(k) + \omega_i\alpha_i(k) + \bar{K}_i\gamma_i(k), \label{eq:filterUpdate}
\end{equation}
where:
\begin{align}
\alpha_i(k) &= \Pi_i(y_o(k) - C_o\hat{{\bf x}}_o^i(k)) \label{eq:alpha}\\
\gamma_i(k) &= \Sigma_i(y_o(k) - C_o\hat{{\bf x}}_o^i(k)). \label{eq:gamma}
\end{align}
Finally, substituting~(\ref{eq:qk}) in~(\ref{eq:alpha}) and~(\ref{eq:gamma}) leads to:
\begin{align*}
\alpha_i(k)  &= \Pi_i( \bar{q}_i(k) + D_i\delta(k-\rho_i)) = \Pi_i \bar{q}_i(k) + \delta(k-\rho_i) \\
\gamma_i(k) &= \Sigma_i( \bar{q}_i(k) + D_i\delta(k-\rho_i)) =  \Sigma_i \bar{q}_i(k), \label{eq:gamma2}
\end{align*}
where we used that $\Pi_i D_i = I_2$ and $\Sigma_i D_i = 0$.

The last two equations verify the desired directional properties for the output residuals, $\alpha_i(k)$ and $\gamma_i(k)$. In fact, we first note that $\gamma_i(k)$ is decoupled from the fault and its convergence to zero is guaranteed by the stability properties of the fault-free filter~(\ref{eq:errorFilter-nofault}) even under a non-zero error initial state estimate~\cite{keller1999fault}.  Moreover, since $\bar{q}_i$ is independent from the fault, we have $\bar{q}_1(k)=\dots=\bar{q}_n(k)$. Finally, as $\bar{q}_i(k)$ approaches zero, $\alpha_i(k)$ converges to the fault $\delta(k-\rho_i)$.

The filter~(\ref{eq:filterUpdate})-(\ref{eq:gamma}) is replicated on board of the observer agent $n$ times, providing the values of $\alpha_i(k)$ and $\gamma_i(k)$ for $i=1,\dots,n$, at all time steps $k>0$. In order to guarantee the correct detection of the fault occurring on the system, three conditions must be satisfied. First, by denoting with $\|\gamma_i(k)\|$ the Euclidean norm of the $i^{\text{th}}$ fault-free residual, trustworthiness of the filter is verified when $\|\gamma_i(k)\|<\epsilon$, where $\epsilon\in\mathbb{R}$ is a small positive tolerance. In addition, by denoting with $\kappa_1,\kappa_2\in\mathbb{R}$, with $\kappa_2<\kappa_1$, two positive fault detection thresholds, uniqueness of the fault is guaranteed when there exists only one residual above the threshold $\kappa_1$, while all other residuals are below the threshold $\kappa_2$, i.e.:  
\begin{equation} \label{eq:detectioncondition}
\begin{cases} 
\exists!\, i,  i=1,\dots,n \, : \|\alpha_{i}(k)\| > \kappa_1 \\
\forall       j=1,\dots,n, \, j \neq i \,: \|\alpha_{j}(k)\| < \kappa_2 \\
\|\gamma_{i}(k)\| < \epsilon
\end{cases}
\Rightarrow 
\begin{cases}
i^* = i \\
k_d = k-\rho_i \\
\delta = \alpha_{i}(k)
\end{cases}
\end{equation}
We refer to the condition in~(\ref{eq:detectioncondition}) as the {\it fault detection condition}.

So far nothing has been said about the choice of $\bar{K}_i$. The directional properties of the residuals are not affected by the particular choices of $\bar{K}_i$, however it represents an additional degree of freedom in the FIF design. In \cite{keller1999fault} $\bar{K}_i$ minimizes the trace of the estimation error covariance matrix. In \cite{sauter2009robust}, $\bar{K}_i$ was designed with respect to the unknown disturbance.

\subsection{Numerical Simulations for the Consensus Dynamics} \label{subsec:cons_result}
We apply the results of the FIF introduced to a team of 9 mobile robots, with interaction topology as in Fig.~\ref{fig:9agentsTopology}.
\begin{figure}[h]
\begin{center}
\includegraphics[width=0.5\columnwidth]{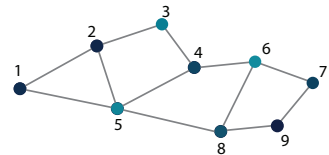}    
\caption{Interaction graph topology for a team of 9 robots.} 
\label{fig:9agentsTopology}
\end{center}
\end{figure}

Starting from random initial positions at time $k=0$, with $\epsilon=0.02$, the robots run the consensus dynamics~(\ref{eq:consensus}). Here, agent $5$ acts as the observer agent, and agent $7$ experiences a fault $\delta=[2, 1]^T$, at $k_d=8$. 

In Fig.~\ref{fig:consensus_residuals} we note the two residuals $\alpha_{i^*}(k)$ (top figure) and $\gamma_{i^*}(k)$ (bottom figure) over time. We observe the components of the fault being correctly estimated in their magnitude (top) and the state estimation error approaching zero from their initial non-zero error (bottom). This confirms the convergence of the state estimate to the real state.
\begin{figure}[h]
\begin{center}
\includegraphics[trim={1.8cm 0 1.8cm 0},width=0.8\columnwidth]{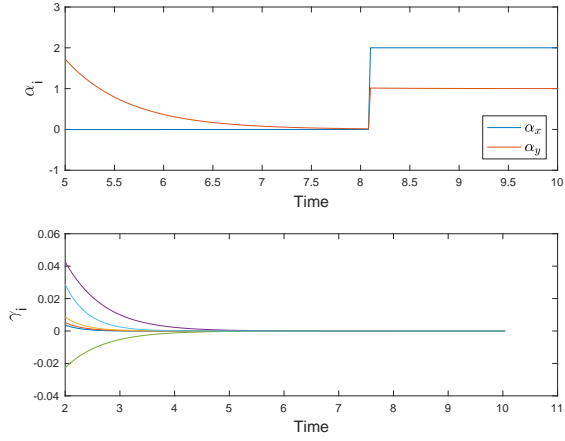}    
\caption{Coupled residuals (top) represent the fault signals. Uncoupled residuals (bottom) represent the estimation error.} 
\label{fig:consensus_residuals}
\end{center}
\end{figure}
\section{Optimal Fault Accommodation}
\label{sec:optimalAccom.}
In the previous section, we discussed the design of the filters bank used by the observer in order to detect a fault occurring in a node of the network. A generic fault was modeled as an exogenous disturbance introduced in the system. In this section, we turn our attention to the leader's accommodating input. In particular, we present an optimal accommodation strategy to be employed by the leader in order to control the robots' centroid, i.e., move the centroid to a predefined recovery position or maintain it to its pre-fault position. Without loss of generality, denoting with $\bar{x}(k)\in\mathbb{R}^{2}$ the centroid of the robots at time step $k$, where $\bar{x}(k)=\frac{1}{n}({\bf 1}_n^T \otimes I_2){\bf x}(k)$, and with $x_f\in\mathbb{R}^2$ the desired final position for $\bar{x}$, leader's objective is to provide the control required such that, under the effect of the fault,
\begin{equation} \label{eq:accobje}
\lim_{k \rightarrow \infty} \bar{x}(k) = x_f.
\end{equation}

Note that, since the only information required by the leader are the fault vector $\delta$, the faulty agent index $i^*$, and the state estimate $\hat{x}$, it is reasonable to assume the leader coinciding with the observer agent. However, if this information can be communicated, this is not required to be necessarily the case. 

\subsection{Control of the Fault and Leader Estimation Filter}
Given the discrete nature of a fault occurring on the system at time step $k_d$, the controlled dynamics can be represented by the following switched controlled consensus:
\begin{equation} \label{eq:leaderUpdate}
\begin{aligned}
x_l(k+1) = x_l(k)& -\epsilon \sum_{j\in \mathcal{N}_i}(x_l(k) - x_j(k)) +  u(k)  \\
u(k) &
\begin{cases}
\neq 0 \, &\text{if} \, k \leq k_d+\rho_{i^*} \\
 = 0     \, &\text{if} \, k > k_d+\rho_{i^*}.
\end{cases}
\end{aligned}
\end{equation}
where $u(k)\in\mathbb{R}^2$ is the leader accommodation control at time $k$. Similarly to what was discussed in the previous section, using~(\ref{eq:leaderUpdate}), the complete post-fault dynamics is:
\begin{equation} \label{eq:consFaultInput}
{\bf x}(k+1) = A{\bf x}(k) + B_l u(k) + \epsilon F_{i^*} \delta, \quad k>k_d+\rho_{i^*},
\end{equation}
where $B_l=\epsilon \bar{B}_l$ and $\bar{B}_l\in\mathbb{R}^{2n \times 2}$ is the control matrix defined by the choice of the leader agent, i.e., $\bar{B}_l = e_l \otimes I_2$. From the post-fault dynamics~(\ref{eq:consFaultInput}), we can define the leader state estimation filter by introducing the control in the filter dynamics~(\ref{eq:genFilter}). Therefore, letting $\hat{{\bf x}}_l(k)$ be the state estimate for the leader agent, we have: 
\begin{multline}\label{eq:leaderFilter}
\begin{aligned}
\hat{{\bf x}}_l(k+1) &= (A-K_{i^*} C_l)\hat{{\bf x}}_l (k) + K_{i^*}\,\hat{y}_l(k) + \\ & \qquad \qquad \qquad \qquad  B_lu(k) +  \epsilon F_{i^*} \delta\\
\hat{y}_l(k) &= C_l\,\hat{{\bf x}}_l(k). 
\end{aligned}
\end{multline}

In order for the leader to achieve the accommodation objective in~(\ref{eq:accobje}), the position of the system's centroid is required. If the state of the system at the initial time is known, from the invariance of $\bar{x}(k)$ under the consensus dynamics, for all $k\leq k_d$:
\begin{equation}\label{eq:centroid}
\bar{x}(k)=\frac{1}{n}({\bf 1}_n^T \otimes I_2){\bf x}(0) = M{\bf x}(0),
\end{equation}
where the definition of $M$ is clear by inspection of~(\ref{eq:centroid}).

Conversely, if the leader does not know the state of the system at the initial time, the position of robots' centroid is also unknown. However, from~(\ref{eq:leaderFilter}) we know that the leader's state estimate $\hat{{\bf x}}_l(k)$ also converges to the centroid of its initial value. Assuming the leader measures its own position $x_l(k)$, it is possible to correct the system state estimate $\hat{{\bf x}}_l$ by the difference between leader's own estimated position, namely $\hat{x}^l_{l}(k)$,  and its measured one. Thus, the position of the centroid at the time of the fault is:
\begin{equation*}
\bar{x}(k_d) = \frac{1}{n}( {\bf 1}_{n}^T \otimes I_2 )( \hat{{\bf x}}_l(k_d) - ({\bf 1}_{n} \otimes I_2)(\hat{x}^l_{l}(k) - x_l(k))).
\end{equation*}

\subsection{Optimal Accommodation Control}
We compute the accommodation control by solving a closed-form receding horizon optimal control problem, in which we assume the system to be completely controllable via the agent $l$~\cite{chipalkatty2013less}. At each time step $k>k_d+\rho_{i^*}$, the solution of the optimal control problem provides a sequence of control inputs $\mathcal{U}(k)=\{u(k+1),\dots,u(k+N-1)\}$, with $N$ being the length of the prediction horizon. The cost to be minimized by the leader is: 
\begin{equation} \label{eq:costFun}
\begin{aligned}
& \underset{u}{\text{min}}
& & \sum_{\tau=0}^{N-1} \| u(k+\tau) \|^2 dt. \\
\end{aligned}
\end{equation}
subject to the system dynamics and the desired $x_f$:
\begin{align}
{\bf x}(k+1) &= A{\bf x}(k) + B_l u(k) + \epsilon F_{i^*} \delta  \label{eq:constr1} \\
M{\bf x}(k+N) &= x_f. \label{eq:constr2}
\end{align}

At the end of the horizon, under the control sequence $\mathcal{U}_k$:
\begin{small}
\begin{equation*}
{\bf x}(k+N) = A^N{\bf x}(k) + \sum_{\tau= 0}^{N-1}A^{N-\tau-1}B_l u(k+\tau) + \epsilon\sum_{\tau = 0}^{N-1}(A^{\tau}) F_{i^*}\delta
\end{equation*} \end{small}
which we rearrange as:
\begin{equation} \label{eq:finalstate}
\sum_{\tau= 0}^{N-1}A^{N-\tau-1}B_l u(k+\tau) = b - A^N{\bf x}(k),
\end{equation}
where $b =  {\bf x}(k+N) -  \epsilon \sum_{\tau = 0}^{N-1}(A^{\tau} ) F_{i^*}\delta$. 

Now, defining the {\it T-Steps Controllability Gramian} \cite{pasqualetti2014controllability} for the discrete time system in~(\ref{eq:consFaultInput}) as $W_c(k,N) = \sum_{\tau= 0}^{N-1}A^{\tau}B_lB_l^T(A^{\tau})^T$, the first element of the control sequence $\mathcal{U}(k)$ to be applied as input to the system is:
\begin{equation}\label{eq:closedform}
u(k)^* = L_k^*M^T(MW_c(k,N)M^T)^{-1}(b - A^N\hat{{\bf x}}_l(k)), 
\end{equation}
where $L_k^* = B_l^T(A^{N-1})^T$.

\subsection{Results}
\label{subsec:resultsControl}
The optimal accommodation strategy is applied to the multi-agent system used in Section~\ref{subsec:cons_result}. A fault $\delta=[2,1]^T$ is applied to agent $8$ at time step $k_d=8$. Once the leader, agent $5$, detects the fault, it initiates the accommodation maneuver. The top of Fig.~\ref{fig:centroidPos} shows the norm of the centroid position $\| \bar{x}(k) \|$ with and without fault accommodation (top), and leader's input components (bottom). As a result of the accommodation strategy, the centroid position remains practically unchanged.

\begin{figure}[h!]
\begin{center}
\includegraphics[trim={1.8cm 0 1.8cm 0},width=0.9\columnwidth]{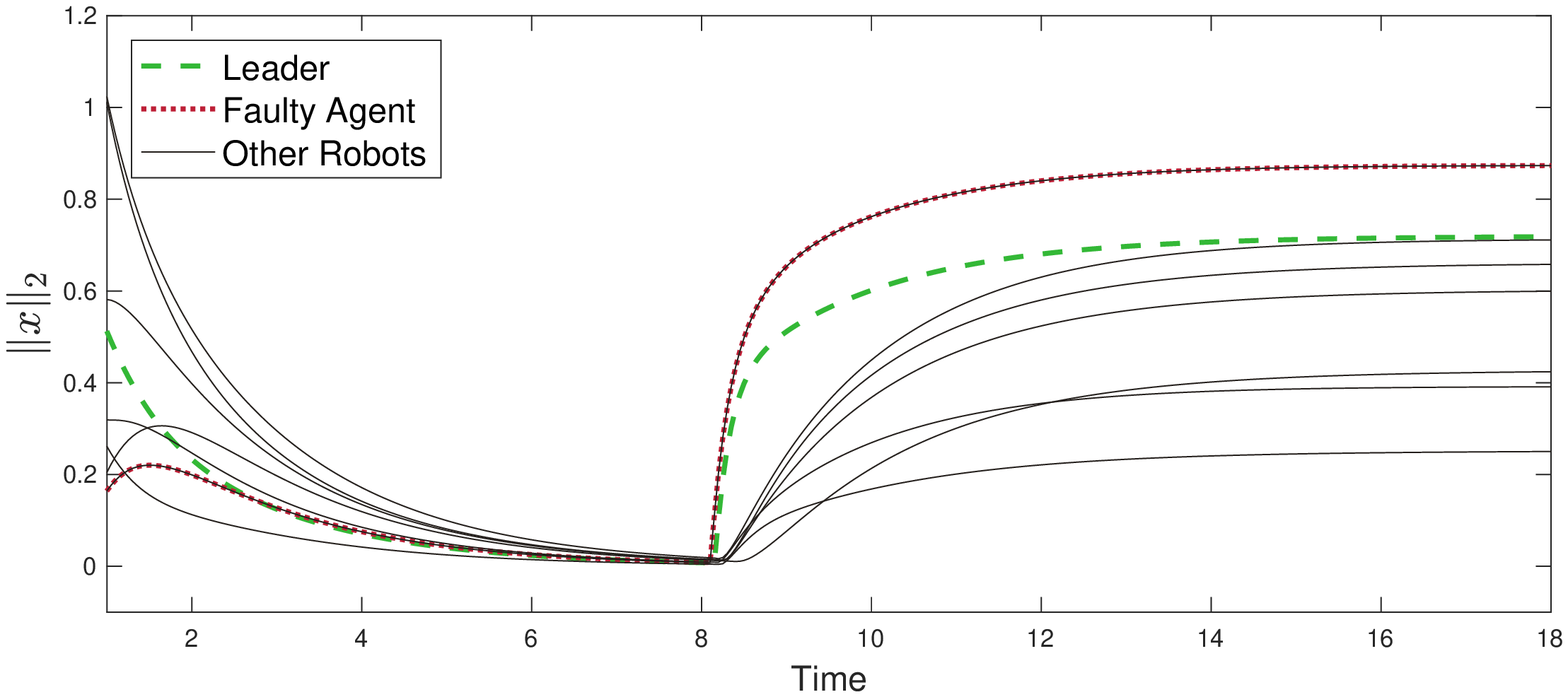}   
\caption{Positions of robots over time. Consensus drives the redezvous until the fault occurs. During accommodation robots reduce their agreement, and reach new equilibrium.\label{fig:statesNorm}} 
\includegraphics[trim={1.8cm 0 1.8cm 0},width=0.9\columnwidth]{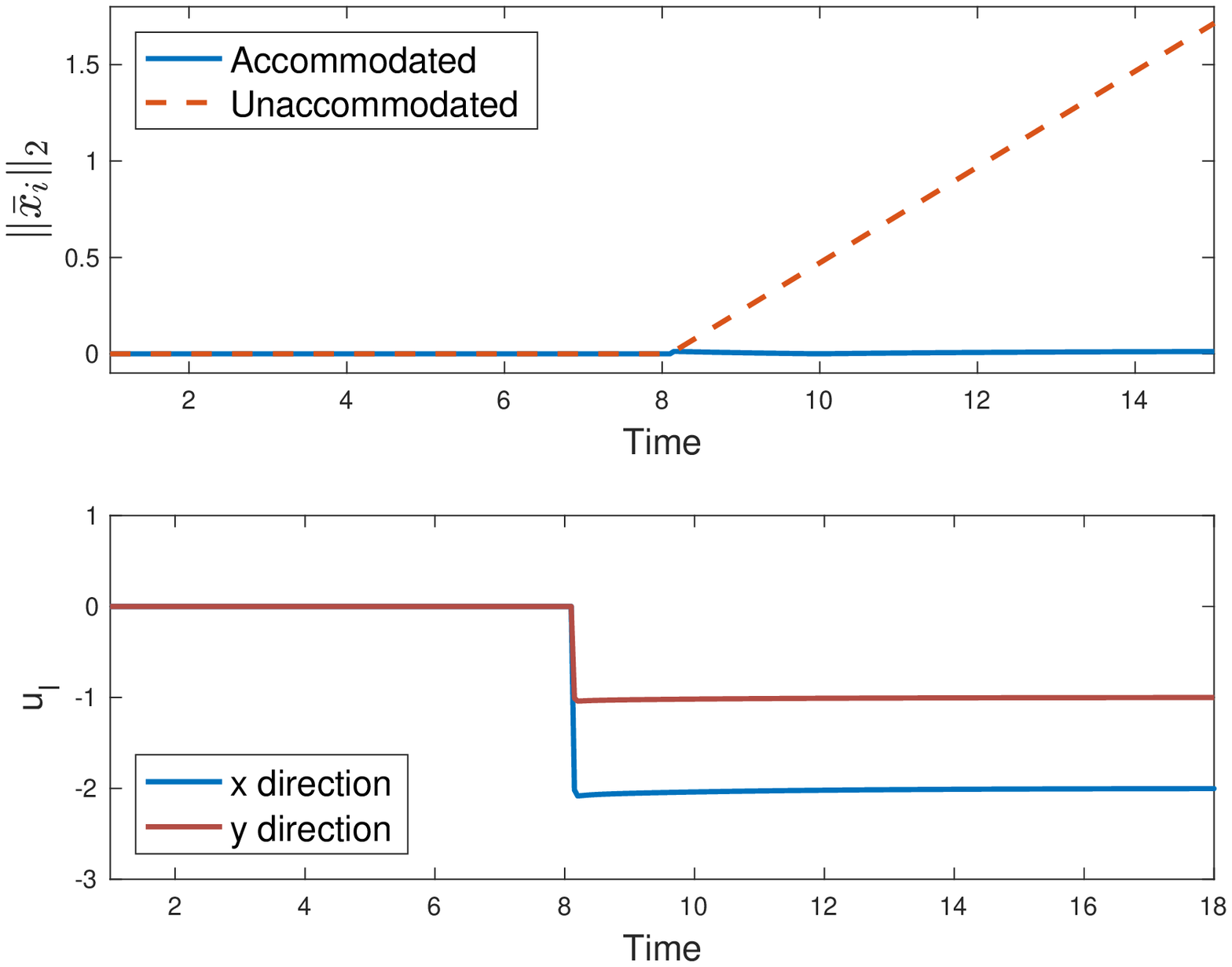}   
\caption{Top - Centroid position with fault accommodation (solid line) and without it (dashed line). During the accommodation phase, the centroid position is practically unchanged. Bottom - Components of accommodation input.\label{fig:centroidPos}} 
\end{center}
\end{figure}

\section{Fault Identification and Accommodation Under Formation Control}
\label{sec:formation}
In this section, we extend the dynamics considered previously to more general scenarios. In particular, we assume that the team of robots runs a consensus-based formation control protocol. Denoting by $d_{ij}\in\mathbb{R}^2$, with $(i,j)\in E$, the desired relative displacement between pairs of neighboring robots, we encode the desired formation in the vector $\phi\in\mathbb{R}^{2n}$, where $\phi_i = \sum_{j \in \mathcal{N}_i} d_{ij}$.

Adding $\phi$ to the update equation in~(\ref{eq:consensusCompact}), we write the formation control protocol as ${\bf x}(k+1) = A{\bf x}(k) + \epsilon\phi$, and the dynamics of the system subject to the fault follow:
\begin{align} \label{eq:switchedsystemformation}
{\bf x}(k+1) &= A{\bf x}(k) + \epsilon F_{i^*} \, \delta \, \nu(k) + \epsilon\phi\\
\nu(k) &=
\begin{cases}
1 \,\, &\text{if} \,\, k \geq k_d \\
0 \,\, &\text{if} \,\, k <    k_d, \\
\end{cases}
\end{align}
where all quantities are the same as in Section~\ref{sec:fif}.

Similarly, it is possible to rewrite the linear filter in~(\ref{eq:genFilter}) for the formation control problem as
\begin{equation} \label{eq:genFilterFormation}
\hat{{\bf x}}_o^i(k+1) = (A-K_i C_o)\hat{{\bf x}}_o^i (k) + K_i\,y_o(k)  + \epsilon\phi.
\end{equation}
Since the formation term $\phi$ does not depend on the state of the system, it can be easily shown that given~(\ref{eq:genFilterFormation}), both~(\ref{eq:errorFilter}) and~(\ref{eq:errorFilter-nofault}) remain unchanged, and consequently, the same estimation gain matrix $K$ computed in~(\ref{eq:K}) still guarantees the desired direction properties for the residuals $\alpha_i$ and $\gamma_i$. By substituting the dynamics of the filters with~(\ref{eq:switchedsystemformation}), the observer applies the same fault detection condition defined in~(\ref{eq:detectioncondition}).

Finally, adding the formation term $\phi$ to the controlled dynamics~(\ref{eq:consFaultInput})
\begin{equation} \label{eq:consFaultInputFormation}
{\bf x}(k+1) = A{\bf x}(k) + B_l u(k) + \epsilon F_{i^*} \delta + \epsilon\phi,
\end{equation}
the final constraint can be rewritten similarly to~(\ref{eq:finalstate}) as:
\begin{equation}
\sum_{\tau= 0}^{N-1}A^{N-\tau-1}B_l u(k+\tau) = b_f - A^N{\bf x}(k),
\end{equation}
where now $b_f =  {\bf x}(k+N) -  \epsilon \sum_{\tau = 0}^{N-1}(A^{\tau} ) (F_{i^*}\delta+\phi)$.

\subsection{Robot Experiments}
\begin{figure*}[h!]
\begin{center}
\begin{subfigure}{.31\textwidth}
  \centering
		\includegraphics[width=0.93\columnwidth]{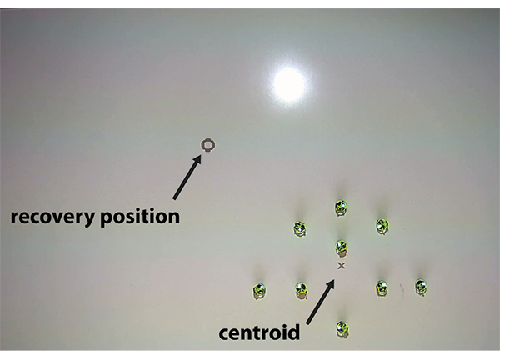} 
  \caption{Formation before the fault.}
\end{subfigure}
\begin{subfigure}{.31\textwidth}
  \centering
		\includegraphics[width=0.93\columnwidth]{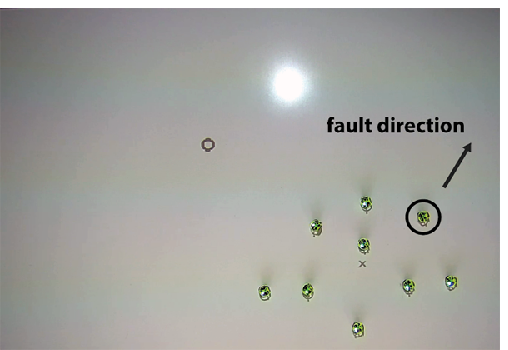} 
  \caption{Fault detected.}
\end{subfigure}	
\begin{subfigure}{.31\textwidth}
  \centering
		\includegraphics[width=0.93\columnwidth]{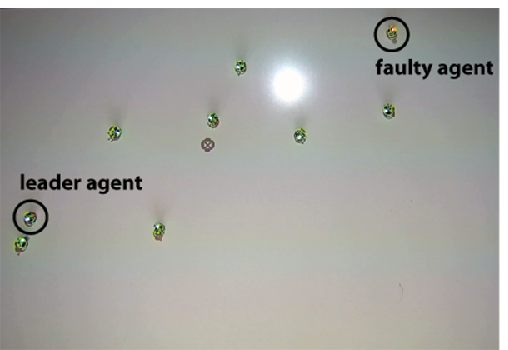}   
  \caption{Centroid moved to recovery position.}
\end{subfigure}	
\caption{Experiment results. Robots performing a formation control task (a). After fault is detected (b), the leader drives the centroid of the team (X in figures) to the recovery position (c). White glare is the light reflaction from a overhead projector. \label{fig:formationPREPOST}} 	
\end{center}
\end{figure*}

Experiments have been performed on the remotely accessible Robotarium~\cite{pickem2017robotarium} platform, with a team of 9 agents performing a formation control protocol. At time step $k_d=45$ a fault vector $\delta=[2,1]^T$ is applied. In this case we assume a {\it recovery position} for the centroid, denoted with a black ring in Fig.~\ref{fig:formationPREPOST} (pictures are taken from an overhead camera). After the fault is detected, the leader compensates for the presence of the faults, and drives the centroid of team (represented by a the black X) to the desired recovery point. In Fig.\ref{fig:formationCentroids} we observe the norm of the centroid moving from the pre-fault position to the desired post-fault value.

\begin{figure}[h!]
	\begin{center}
		\includegraphics[trim={2.4cm 0 2.4cm 0},width=0.9\columnwidth]{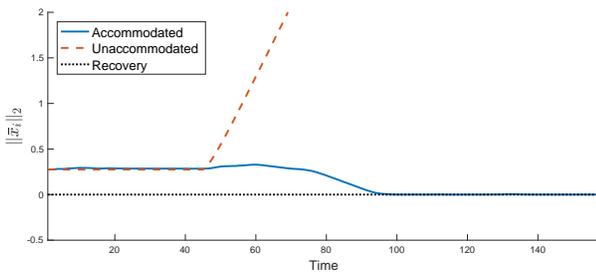}   
		\caption{Centroid position over time. After the fault time $k_d=45$, the accommodation maneuver moves the centroid from its initial position to the desired recovery value, located at 0. Dashed line is a truncated ramp and represents the position of the centroid if the case without accommodation. } 
		\label{fig:formationCentroids}
	\end{center}
\end{figure}

\section{Conclusion}
\label{sec:conc}
Consensus-based protocols in multi-agents systems are highly vulnerable to exogenous disturbances, such as faults. In this paper, a fault identification and accommodation strategy for a static networked  multi-agent robotic system is proposed. Under a linear agreement and formation control, the proposed filter individuates a faulty agent anywhere in the network and estimates the entity of the disturbance introduced. After the fault is detected, an optimal accommodation strategy is employed by a leader in order to control the robots' centroid, and move it to an arbitrary position. 

                                                   
\bibliographystyle{IEEEtran}
\bibliography{IEEEabrv,bibfile.bib}
%
\end{document}